\documentclass[twocolumn,
showpacs,preprintnumbers,amsmath,amssymb]{revtex4}

\usepackage{graphicx}
\usepackage{amsthm}

\newtheorem{Thm}{Theorem}

\newtheorem{Lem}[Thm]{Lemma}
\newcommand{\bra}[1]{{\left\langle #1 \right|}}
\newcommand{\ket}[1]{{\left| #1 \right\rangle}}

\newcommand{\T}{\mbox{$\mathrm{tr}$}}

\newcommand{\N}{\mbox{$\mathcal{N}$}}
\newcommand{\I}{\mbox{$\mathcal{I}$}}
\newcommand{\E}{\mbox{$\mathcal{E}$}}
\begin{document}

\title{Separable states to distribute entanglement}

\author{Jeonghoon Park}\email{zucht@khu.ac.kr}
\affiliation{
 Department of Mathematics and Research Institute for Basic Sciences,
 Kyung Hee University, Seoul 130-701, Korea
}
\author{Soojoon Lee}\email{level@khu.ac.kr}
\affiliation{
 Department of Mathematics and Research Institute for Basic Sciences,
 Kyung Hee University, Seoul 130-701, Korea
}

\date{\today}

\begin{abstract}
It was shown that
two distant particles can be entangled
by sending a third particle never entangled with the other two
[T. S. Cubitt {\em et al.}, Phys. Rev. Lett. {\bf 91}, 037902 (2003)].
In this paper,
we investigate a class of three-qubit separable states to distribute entanglement
by the same way,
and calculate the maximal amount of entanglement
which two particles of separable states in the class can have
after applying the way.
\end{abstract}

\maketitle

Entanglement is one of the most crucial resources
in quantum information processing.
Thus, in order to achieve quantum effects in information processing,
it is important to create entanglement between two or more particles.

It would seem that
the only way to entangle two particles is
either to apply global quantum operation
or
to send an ancillary particle entangled with one particle
to the other particle.
However, it was shown that
two particles can become entangled
without the ancillary particle ever becoming entangled~\cite{CVDC}.
In other words,
there exist a separable state to distribute entanglement.

In Ref.~\cite{CVDC}, the authors gave an explicit procedure,
which we call the Cubitt-Verstraete-D\"{u}r-Cirac (CVDC) procedure, as follows:
(i)~Alice and Bob prepare the (separable) state
$\sigma'_{abc}\equiv \mathrm{CNOT}_{ac}\rho'_{abc}\mathrm{CNOT}_{ac}^\dagger$,
where $\mathrm{CNOT}_{ac}$ is the controlled NOT (CNOT) operation
on particles $a$ as the control qubit and $c$ as the target qubit,
and
\begin{equation}
\rho'_{abc}\equiv
\frac{1}{3}\ket{\Psi^+_0}\bra{\Psi^+_0}
+\frac{1}{6}\sum_{j=1,3}(\ket{\Psi^+_j}\bra{\Psi^+_j}+\ket{\Psi^-_j}\bra{\Psi^-_j})
\label{eq:original}
\end{equation}
with $\ket{\Psi^{\pm}_j}=(\ket{j}\ket{0}\pm\ket{3-j}\ket{1})/\sqrt{2}$.
(ii)~Alice applies $\mathrm{CNOT}_{ac}$ to the state $\sigma'_{abc}$,
producing the state $\rho'_{abc}$,
and then sends the particle $c$ to Bob.
(iii)~Bob applies $\mathrm{CNOT}_{bc}$ to the state $\rho'_{abc}$,
resulting in the state
\[
\tau'_{abc}\equiv \frac{1}{3}\ket{\phi^+}_{ab}\bra{\phi^+}\otimes\ket{0}_c\bra{0}
+\frac{1}{6}\I_{ab}\otimes\ket{1}_c\bra{1},
\]
where $\ket{\phi^+}=(\ket{00}+\ket{11})/\sqrt{2}$.
(iv)~Bob measures the particle $c$ in the computational basis.
Then the maximally entangled state $\ket{\phi^+}_{ab}$
can be extracted between the particles $ab$ with probability $1/3$.

In the CVDC procedure,
the ancillary qubit $c$ always remains separable with the particles $ab$.
Nevertheless, after performing the procedure,
entanglement can be probabilistically distributed between the particles $ab$.
Replacing the step~(iv) by the following step~(iv$'$),
a deterministic CVDC procedure can be obtained:
(iv$'$) Bob applies a local completely-positive trace-preserving (CPTP) map to the particles $bc$,
defined by $\mathcal{E}_{bc}(\rho)=\sum_j O_{bc}^{(j)}\rho O_{bc}^{(j)\dagger}$
with Kraus operators
$O_{bc}^{(1)}=\I_b\otimes\ket{0}_c\bra{0}$,
$O_{bc}^{(2)}=\ket{0}_b\bra{0}\otimes\ket{1}_c\bra{1}$, and
$O_{bc}^{(3)}=\ket{0}_b\bra{1}\otimes\ket{1}_c\bra{1}$.
Tracing out the ancillary qubit $c$,
the resulting state $\T_c\left(\E_{bc}\left(\tau'_{abc}\right)\right)$
becomes entangled, since the state has non-positive partial transpose.

From the above procedures,
some separable state can probabilistically/deterministically distribute entanglement.
Then one could naturally ask
what kind of separable states can distribute entanglement
and how much entanglement can be obtained from those separable states
by the probabilistic/deterministic CVDC procedure.
In this paper, we examine a class of three-qubit separable states
to distribute entanglement between the two qubits
after performing the CVDC procedure,
and calculate the (maximal) amount of entanglement
distributed between the two qubits for three-qubit states in the class.

As seen in the above procedure,
the state $\rho'_{abc}$ in Eq.~(\ref{eq:original})
is clearly an element in the family of three-qubit states with four parameters
presented by D\"{u}r {\em et al.}~\cite{DCT},
\begin{equation}
\rho_{abc}=\sum_{\sigma=\pm}\lambda^\sigma_0\ket{\Psi^\sigma_0}\bra{\Psi^\sigma_0}
+\sum^3_{j=1}\lambda_j(\ket{\Psi^+_j}\bra{\Psi^+_j}+\ket{\Psi^-_j}\bra{\Psi^-_j}),
\label{eq:Dur_state}
\end{equation}
where $0\leq\lambda^\sigma_0,\lambda_j\leq1$,
$\lambda^+_0\geq\lambda^-_0,\lambda_1,\lambda_2,\lambda_3$,
and $\lambda_0^+ + \lambda_0^- + 2\sum_{j} \lambda_j =1$.
We now let the state $\rho'_{abc}$ be replaced by the states $\rho_{abc}$ in the class,
and let $\sigma_{abc}\equiv \mathrm{CNOT}_{ac}\rho_{abc}\mathrm{CNOT}_{ac}^\dagger$
and $\tau_{abc}\equiv \mathrm{CNOT}_{bc}\rho_{abc}\mathrm{CNOT}_{bc}^\dagger$.
In order to obtain a class of separable states to distribute entanglement
from the CVDC procedure,
the following four conditions are required:
\begin{itemize}
\item[(a)] $\sigma_{abc}$ is separable.
\item[(b)] $\T_c\left(\rho_{abc}\right)$ is separable.
\item[(c)] $\rho_{abc}$ and $\tau_{abc}$
are separable between two parties $ab$-$c$.
\item[(d)] Measuring the particle $c$ on the states $\sigma_{abc}$ and $\rho_{abc}$ in the computational basis,
the resulting states on the particles $ab$ are separable.
But, measuring the particle $c$ of $\tau_{abc}$ in the computational basis,
the resulting state on the particles $ab$ has entanglement with nonzero probability.
\end{itemize}

According to the above conditions,
we investigate three-qubit separable states in the family to distribute entanglement.
First we consider the condition for $\sigma_{abc}$.
It can be readily obtained that
\begin{equation}
\sigma_{abc}=
\sigma_{ab}^{(0)}\otimes{\ket{0}_c\bra{0}}
+\sigma_{ab}^{(1)}\otimes{\ket{1}_c\bra{1}},
\label{eq:sigma}
\end{equation}
where
\begin{equation*}
\sigma_{ab}^{(0)}=
\begin{pmatrix}
\delta/2 & 0 & 0 & \Delta/2 \\
0 & \lambda_1 & 0 & 0 \\
0 & 0 & \lambda_1 & 0 \\
\Delta/2 & 0 & 0 & \delta/2
\end{pmatrix}
\end{equation*}
with $\Delta\equiv\lambda^+_0-\lambda^-_0$ and $\delta \equiv {\lambda^+_0+\lambda^-_0}$,
and
\begin{equation*}
\sigma_{ab}^{(1)}=
\begin{pmatrix}
\lambda_3 & 0 & 0 & 0 \\
0 & \lambda_2 & 0 & 0 \\
0 & 0 & \lambda_2 & 0 \\
0 & 0 & 0 & \lambda_3
\end{pmatrix}.
\label{eq:sigma1}
\end{equation*}
Thus, by the condition~(d), $\sigma_{ab}^{(0)}$ should be separable,
and hence we obtain the inequality
\begin{equation}
2\lambda_1\ge \Delta.
\label{eq:lambda1}
\end{equation}
Then since $\sigma_{ab}^{(1)}$ is clearly separable,
the condition~(a) is also satisfied.

We now investigate the condition for $\rho_{abc}$.
Since the state $\T_c\left(\rho_{abc}\right)$ is trivially separable,
the condition~(b) naturally holds.
Furthermore, since it can be straightforwardly shown that
the resulting states on $ab$ after the measurement of the particle $c$ of $\rho_{abc}$
are separable,
the condition for $\rho_{abc}$ in~(d) is also satisfied.
However, by the condition~(c),
the partial transpose of the state $\rho_{abc}$ with respect to the particle $c$
must be positive, that is,
\begin{equation}
\rho_{abc}^{T_c}=
\begin{pmatrix}
\delta/2 & 0 & 0 & 0 & 0 & 0 & 0 & 0 \\
0 & \lambda_3 & 0 & 0 & 0 & 0 & \Delta/2 & 0 \\
0 & 0 & \lambda_1 & 0 & 0 & 0 & 0 & 0 \\
0 & 0 & 0 & \lambda_2 & 0 & 0 & 0 & 0 \\
0 & 0 & 0 & 0 & \lambda_2 & 0 & 0 & 0 \\
0 & 0 & 0 & 0 & 0 & \lambda_1 & 0 & 0 \\
0 & \Delta/2 & 0 & 0 & 0 & 0 & \lambda_3 & 0 \\
0 & 0 & 0 & 0 & 0 & 0 & 0 & \delta/2
\end{pmatrix}
\ge 0,
\label{eq:rho_Tc}
\end{equation}
and we thus obtain the inequality
\begin{equation}
2\lambda_3\ge \Delta.
\label{eq:lambda3}
\end{equation}

We finally consider the condition for $\tau_{abc}$.
It follows from direct calculations that
\begin{equation}
\tau_{abc}=
\tau_{ab}^{(0)}\otimes{\ket{0}_c\bra{0}}
+\tau_{ab}^{(1)}\otimes{\ket{1}_c\bra{1}},
\label{eq:tau}
\end{equation}
where $\tau_{ab}^{(0)}$ and $\tau_{ab}^{(1)}$ are
\begin{eqnarray}
\begin{pmatrix}
\delta/2 & 0 & 0 & \Delta/2 \\
0 & \lambda_2 & 0 & 0 \\
0 & 0 & \lambda_2 & 0 \\
\Delta/2 & 0 & 0 & \delta/2
\end{pmatrix}~\mathrm{and}~
\begin{pmatrix}
\lambda_3 & 0 & 0 & 0 \\
0 & \lambda_1 & 0 & 0 \\
0 & 0 & \lambda_1 & 0 \\
0 & 0 & 0 & \lambda_3
\end{pmatrix},
\label{eq:tau01}
\end{eqnarray}
respectively.
Then the condition~(c) clearly holds.
Since the state~$\tau_{ab}^{(1)}$ is separable,
the state $\tau_{ab}^{(0)}$ should be entangled by the condition~(d).
Hence we obtain the following inequality:
\begin{equation}
2\lambda_2< \Delta.
\label{eq:lambda2}
\end{equation}

By the inequalities~(\ref{eq:lambda1}), (\ref{eq:lambda3}), and (\ref{eq:lambda2}),
we can have the following theorem.
\begin{Thm}\label{Thm:lambdas}
Separable states $\sigma_{abc}$ in Eq.~(\ref{eq:sigma})
can (probabilistically) distribute entanglement
between the particles $ab$ by the CVDC procedure
if and only if $2\lambda_1\ge \Delta$, $2\lambda_2< \Delta$, and $2\lambda_3\ge \Delta$.
\end{Thm}
We note that the converse statement of Theorem~\ref{Thm:lambdas} is true
since $\sigma_{ab}^{(0)}$ and $\tau_{ab}^{(0)}$ are two-qubit states,
and the positive partial transpose of $\rho_{abc}$ implies its separability~\cite{DCT}.

We remark that
the state $\T_c\left(\tau_{abc}\right)$ is not entangled.
However, as in the deterministic CVDC procedure,
applying the same local CPTP map $\E_{bc}$
to the particles $bc$ of the state $\tau_{abc}$,
the final state $\T_c\left(\E_{bc}\left(\tau_{abc}\right)\right)$
can be entangled if $4\lambda_2\left(\lambda_1+\lambda_2+\lambda_3\right)<\Delta^2$,
since the state $\T_c\left(\E_{bc}\left(\tau_{abc}\right)\right)$ is
\begin{equation}
\begin{pmatrix}
\delta/2+\lambda_1+\lambda_3 & 0 &  0 & \Delta/2\\
0  & \lambda_2 & 0 & 0 \\
0 & 0 & \lambda_1+\lambda_2+\lambda_3 & 0 \\
\Delta/2 & 0 & 0 & \delta/{2}
\end{pmatrix}.
\label{eq:TcEbc_tau}
\end{equation}

Let $\mathcal{S}$ be a class of three-qubit separable states of the form $\sigma_{abc}$ in Eq.~(\ref{eq:sigma})
satisfying $2\lambda_1\ge \Delta$, $2\lambda_2< \Delta$, and $2\lambda_3\ge \Delta$.
Then all separable states in the class $\mathcal{S}$ can probabilistically distribute entanglement
by the CVDC procedure,
and the following bounds on the parameter $\Delta$
for separable states $\sigma_{abc}$ in $\mathcal{S}$
can be obtained from Theorem~\ref{Thm:lambdas}.
\begin{Thm}\label{Thm:Delta}
Let $\sigma_{abc}$ be a three-qubit separable state in $\mathcal{S}$.
Then the parameter $\Delta$ of $\sigma_{abc}$ satisfies the inequality $0<\Delta\le 1/3$.
Furthermore, for each $0<\Delta\le 1/3$,
there exists a state 
in $\mathcal{S}$ with the parameter $\Delta$.
This implies that
the inequality provides us with the tight bounds on the parameter $\Delta$.
\end{Thm}
\begin{proof}
Since $\sigma_{abc}\in \mathcal{S}$, it follows from Theorem~\ref{Thm:lambdas}
that $2\lambda_1\ge \Delta$, $2\lambda_2< \Delta$, and $2\lambda_3\ge \Delta$.
Thus, $\Delta$ must be nonzero, and
\begin{eqnarray}
\Delta&\le&\frac{1}{3}\left(\lambda^+_0+2\lambda_1+2\lambda_3\right)\nonumber\\
&\le& \frac{1}{3}\left(\lambda_0^+ + \lambda_0^- + 2\sum_{j} \lambda_j \right)=\frac{1}{3}.
\label{eq:Delta_ineq}
\end{eqnarray}
We now let $\Delta$ be a real number such that $0<\Delta\le 1/3$.
Then taking $\lambda^+_0=(1-\Delta)/2$, $\lambda^-_0=\lambda^+_0-\Delta$,
$\lambda_1=\lambda_3=\Delta/2$, and $\lambda_2=0$,
the state $\sigma_{abc}$ becomes
\begin{equation}
\sigma_{\Delta}\equiv
\sigma_{\Delta}^{(0)}\otimes{\ket{0}\bra{0}}
+\sigma_{\Delta}^{(1)}\otimes{\ket{1}\bra{1}},
\label{eq:sigma_Delta}
\end{equation}
where
\begin{equation*}
\sigma_{\Delta}^{(0)}=
\begin{pmatrix}
(1-3\Delta)/4 & 0 & 0 & \Delta/2 \\
0 & \Delta/2 & 0 & 0 \\
0 & 0 & \Delta/2 & 0 \\
\Delta/2 & 0 & 0 & (1-3\Delta)/4
\end{pmatrix}
\end{equation*}
and
\begin{equation*}
\sigma_{\Delta}^{(1)}=
\begin{pmatrix}
\Delta/2 & 0 & 0 & 0 \\
0 & 0 & 0 & 0 \\
0 & 0 & 0 & 0 \\
0 & 0 & 0 & \Delta/2
\end{pmatrix}.
\end{equation*}
By Theorem~\ref{Thm:lambdas},
the state $\sigma_{\Delta}$ clearly becomes an element in $\mathcal{S}$.
Hence, there exists a three-qubit state in $\mathcal{S}$ with the parameter $\Delta$.
\end{proof}
We consider the state $\sigma_{abc}\in \mathcal{S}$ in Eq.~(\ref{eq:sigma}) with parameter $\Delta=1/3$.
Then since $2\lambda_1\ge\Delta=1/3$ and $2\lambda_3\ge\Delta=1/3$,
\begin{equation}
1=2\lambda_0^+ -\Delta+2\sum_{j=1}^3\lambda_j\ge 2\lambda_0^+ +2\lambda_2+\frac{1}{3}.
\label{eq:trace_one}
\end{equation}
This implies the inequality $1/3\ge \lambda_0^+ +\lambda_2$.
Since
\[
\lambda_0^+ +\lambda_2\ge \lambda_0^+=\Delta + \lambda_0^- = \frac{1}{3} + \lambda_0^-\ge \frac{1}{3},
\]
we have $\lambda_0^- = 0$, and hence $\lambda_0^+=\Delta=1/3$.
Then, from the inequality~(\ref{eq:trace_one}),
we can obtain the inequality $1\ge 1+ 2\lambda_2$, and hence $\lambda_2$ becomes zero.
By the left-hand side of the inequality~(\ref{eq:trace_one}),
we readily obtain that $\lambda_1=\lambda_3=1/6$.
Therefore, we can see that the state $\sigma_{abc}\in \mathcal{S}$ with $\Delta=1/3$
uniquely becomes $\sigma'_{abc}$ in the original CVDC procedure.

For the state $\sigma_\Delta$ with one parameter $\Delta$ in Eq.~(\ref{eq:sigma_Delta}),
let
$\tau_{\Delta}\equiv
\mathrm{CNOT}_{bc}\mathrm{CNOT}_{ac}\sigma_{\Delta}\mathrm{CNOT}_{ac}^\dagger\mathrm{CNOT}_{bc}^\dagger$.
Then we note that
the state $\T_c\left(\E_{bc}\left(\tau_{\Delta}\right)\right)$ becomes entangled,
and hence the one-parameter state $\sigma_\Delta$
can deterministically distribute entanglement between two particles $ab$
by the CVDC procedure with the local CPTP map $\E_{bc}$.

We now calculate the amount of entanglement on the particles $ab$
attainable by the CVDC procedure.
Let $p_e$ be the probability to obtain nonzero entanglement
after the measurement of the particle $c$ in the state $\tau_{abc}$.
Then we have $p_e=\T(\tau_{ab}^{(0)})=\delta+2\lambda_2=1-2\lambda_1-2\lambda_3$.
Let $\N$ be the negativity~\cite{LKPL,VidalW,LCOK} as an entanglement measure defined by
\begin{equation}
\N(\rho)\equiv \left\|\rho^{T_a}\right\|-1
\end{equation}
for a two-qubit state $\rho$,
where $\|\cdot\|$ is the trace norm.
Then we can readily obtain $\N(\tau_{ab}^{(0)})=\Delta-2\lambda_2$.
Thus, the average amount of entanglement obtainable from the state $\tau_{abc}$
after performing the measurement of the particle $c$ is
\begin{equation}
p_e \cdot \N(\tau_{ab}^{(0)}/p_e)=\N(\tau_{ab}^{(0)})=\Delta-2\lambda_2\le\Delta,
\label{eq:average_entanglement}
\end{equation}
where $\Delta$ is the average amount of entanglement obtainable from the state $\tau_{\Delta}$
after measuring the particle $c$.
Hence we obtain the following theorem.
\begin{Thm}\label{Thm:sigma_Delta}
For each $0<\Delta\le 1/3$,
the state $\sigma_\Delta$ has the maximal average amount of entanglement
among all states in $\mathcal{S}$ with the parameter $\Delta$.
\end{Thm}
By Theorem~\ref{Thm:sigma_Delta}, 
it is known that the maximal amount of entanglement
after applying the CVDC procedure is $1/3$ when $\Delta=1/3$,
which is the original case.

We finally consider the state $\T_c\left(\E_{bc}\left(\tau_{abc}\right)\right)$.
Then its negativity value can be straightforwardly obtained as follows:
\begin{equation}
\max\left\{\sqrt{(\lambda_1+\lambda_3)^2+\Delta^2}-(\lambda_1+2\lambda_2+\lambda_3),0\right\}.
\label{eq:final_negativity}
\end{equation}
In the case of the deterministic CVDC procedure,
we have a result similar to Theorem~\ref{Thm:sigma_Delta}.
In order to prove the result,
we begin with the following simple lemma.
\begin{Lem}\label{Lem:sigma_Delta}
Let $\Delta$ be a fixed number such that $0<\Delta\le1/3$,
and let $f$ be a function on a closed interval $[\Delta, (1-\Delta)/2]$
defined by $f(x)\equiv \sqrt{x^2+\Delta^2}-x$.
Then the function $f$ has the maximal value $(\sqrt{2}-1)\Delta$ at $x=\Delta$.
\end{Lem}
The above lemma is trivial,
since the function $f$ has no critical points in its domain,
and
\[
f(\Delta)=(\sqrt{2}-1)\Delta\ge f((1-\Delta)/2).
\]
Since for each state $\sigma_\Delta$,
the value in Eq.~(\ref{eq:final_negativity}) becomes {$(\sqrt{2}-1)\Delta$},
we clearly have the following theorem.
\begin{Thm}\label{Thm:sigma_Delta2}
For each $0<\Delta\le 1/3$,
the state $\sigma_\Delta$ has the maximal amount of entanglement
among all states in $\mathcal{S}$ with the parameter $\Delta$,
after performing the deterministic CVDC procedure.
\end{Thm}
\begin{proof}
For $\sigma_{abc}$ in $\mathcal{S}$, 
$\Delta\le \lambda_1+\lambda_3\le(1-\Delta)/2$.
It follows from Lemma~\ref{Lem:sigma_Delta}
that $f(\lambda_1+\lambda_3)\le f(\Delta)$.
Since the value in Eq.~(\ref{eq:final_negativity}) is 
less than or equal to $f(\lambda_1+\lambda_3)$ and 
$f(\Delta)$ is the amount of entanglement of $\T_c\left(\E_{bc}\left(\tau_{\Delta}\right)\right)$,
we complete the proof.
\end{proof}
Hence the maximal amount of entanglement
after applying the CVDC procedure with the local CPTP map $\E_{bc}$
is $(\sqrt{2}-1)/3$ when $\Delta=1/3$,
as in the probabilistic case.

In summary,
we have investigated a class of three-qubit separable states to distribute entanglement
by the CVDC procedure,
and have calculated the (maximal) amount of entanglement
which two particles of separable states in the class can have
after applying the procedure.

This work has been supported by
Basic Science Research Program
through the National Research Foundation of Korea (NRF)
funded by the Ministry of Education, Science and Technology (Grant No.~2009-0076578).


\end{document}